\documentclass[DIV12,11pt]{scrartcl}

\usepackage{amsthm,amssymb,amsmath,tikz}
\usetikzlibrary{arrows,positioning,fit,shapes,decorations.pathmorphing}
\usepackage[noblocks]{authblk}
\usepackage{graphicx}

\setkomafont{title}{\normalfont \LARGE \bfseries}
\setkomafont{section}{\normalfont \Large \bfseries \boldmath}
\setkomafont{subsection}{\normalfont \large \bfseries}
\setkomafont{subsubsection}{\normalfont \normalsize \bfseries}
\setkomafont{descriptionlabel}{\itshape}
\setkomafont{paragraph}{\normalfont \bfseries \boldmath}

\newtheorem{theorem}{Theorem}
\newtheorem{lemma}[theorem]{Lemma}
\newtheorem{corollary}[theorem]{Corollary}

\title{Rank Aggregation: New Bounds for MCx}

\begin{document}
\author[1]{Daniel Freund}
\author[2]{David P.\ Williamson\thanks{The second author was supported in part by NSF grant CCF-1115256.}}
\affil[1]{Center for Applied Mathematics, Cornell University\\
  Ithaca, NY, USA\\
  \texttt{df365@cornell.edu}}
\affil[2]{School of Operations Research and Information Engineering, Cornell University\\
Ithaca, NY, USA\\
\texttt{dpw@cs.cornell.edu}}
\maketitle

\begin{abstract}
The rank aggregation problem has received significant recent attention within the computer science community. Its applications today range far beyond the original aim of building metasearch engines to problems in machine learning, recommendation systems and more. Several algorithms have been proposed for these problems, and in many cases approximation guarantees have been proven for them.   However, it is also known that some Markov chain based algorithms (MC1, MC2, MC3, MC4) perform extremely well in practice, yet had no known performance guarantees.  We prove supra-constant lower bounds on approximation guarantees for all of them. We also raise the lower bound for sorting by Copeland score from $\frac{3}{2}$ to $2$ and prove an upper bound of $11$, before showing that in particular ways, MC4 can nevertheless be seen as a generalization of Copeland score.
 \end{abstract}

\section{Introduction}

We consider the problem of ranking a set of elements given a number of input rankings of the set.  The goal is to find a ranking that minimizes the deviation from the input set of rankings.  This problem has been well-studied in voting and social choice theory, but recent interest of the computer science community was started by a paper of Dwork, Kumar, Naor, and Sivakumar \cite{Dwork}, who called the problem the {\em rank aggregation problem}.  Dwork et al.\ studied the problem in the context of metasearch; that is, aggregating the results of multiple search engines into a single ranking.

To be more precise, we suppose that the input rankings are given as permutations $\pi_i: [n] \rightarrow [n]$ for $i=1,\ldots,m$; the goal is to find another permutation $\sigma: [n] \rightarrow [n]$ to minimize $\sum_{i=1}^m K(\sigma, \pi_i)$, where $K$ is the Kendall distance, a measure of the distance between the two permutations $\sigma$ and $\pi_i$.  From the input permutations, we can obtain a complete weighted directed graph $G=(V,A)$ where $V=[n]$ and the weight $w_{ij}$ of arc $(i,j)$ is the fraction of permutations that rank $i$ ahead of $j$.

Quite a number of different algorithms have been proposed for performing rank aggregation.  In Borda scoring, the vertices of the graph are sorted in nondecreasing order of their weighted indegree.  The Copeland score of an element $i$ is the number of $j$ such that $i$ is ranked above $j$ in at least half the rankings; another method is to rank by nondecreasing Copeland score. Ailon, Charikar, and Newman \cite{Newman} gave an algorithm motivated by the Quicksort algorithm in which at each step, a random element $i$ is chosen as a ``pivot''; remaining elements $j$ are then ordered before or after the pivot  $i$ depending on whether $w_{ji}$ or $w_{ij}$ is greater, and then the elements before and after the pivot are recursively ordered.  Deterministic variants of these pivoting algorithms were given by van Zuylen and Williamson \cite{Anke2}.  Motivated by the PageRank algorithm \cite{PageRank}, Dwork et al.\ proposed four Markov chain-based methods to obtain an overall ranking; they called these algorithms MC1, MC2, MC3, and MC4.  Each chain is on the set $[n]$ with different transition probabilities specified.  To obtain an ordering, the stationary probabilities of the chain are computed, and the elements are then ordered in nonincreasing order of stationary probability.  In computational experiments, Dwork et al.\ showed that MC4 was especially effective in finding a near-optimal solution to the rank aggregation problem on data drawn from metasearch.  These results were further validated in experiments of Schalekamp and van Zuylen \cite{Anke} on other data sets, who also found that MC4 outperformed other methods, including the Borda algorithm and several variants of the pivoting algorithms.  Dwork et al.\ argue that MC2 and MC3 generalize the Borda scoring method, and that MC4 is a generalization of the Copeland scoring method.

Many of these algorithms have been studied from an approximation algorithms perspective; we say that an algorithm for the rank aggregation algorithm is an $\alpha$-approximation algorithm if it runs in polynomial time and on any input returns a solution within a factor of $\alpha$ of the optimal solution for that input.  The factor $\alpha$ is sometimes called the performance guarantee or approximation factor of the algorithm.  Ailon et al.\ give a $\frac{4}{3}$-approximation algorithm via an LP-based variant of their pivoting algorithm.  Coppersmith, Fleischer, and Rudra \cite{CFR} show that Borda scoring is a 5-approximation algorithm.  Kenyon-Mathieu and Schudy \cite{Claire} show that for any fixed $\epsilon > 0$, they can give a $(1+\epsilon)$-approximation algorithm for the problem, although the running time of the algorithm is doubly exponential in $1/\epsilon$.
However, no bound on the approximation factor of the Markov chain-based algorithms is known, despite their excellent performance in practice.  Determining the performance guarantee of these algorithms was raised as an open question by Schalekamp and van Zuylen; they give an example to show that it must be at least $\frac{3}{2}$.

In this paper, we resolve the question of Schalekamp and van Zuylen by showing -- somewhat surprisingly -- that the Markov chain-based algorithms all have bad performance guarantees.  In particular, for MC1, MC2, and MC3, we give an example to show that these algorithms have an arbitrarily bad performance guarantee on a set of just three elements.  For the MC4 algorithm, we give a family of examples that show that it does not have any sublinear performance guarantee.

We also study the performance of the Copeland algorithm.  Schalekamp and van Zuylen show that its approximation factor cannot be better than $\frac{3}{2}$; we improve this to show that it cannot be better than 2.  By drawing on the proofs of Coppersmith et al.\ for the Borda scoring algorithm, we are able to show that Copeland is a 11-approximation algorithm.  Fagin, Kumar, Mahdian, Sivakumar, and Vee \cite{FKMSV}, in an unpublished manuscript, have shown that Copeland is a 4-approximation algorithm.  We include our result so that there is some published bound on Copeland in the literature.  We also show a strong connection between particular implementations of MC4 given in the literature, and Copeland scoring.

The rest of the paper is structured as follows.  We introduce notation and give the precise definition of the rank aggregation problem in Section \ref{sec:notation}, as well as the definitions of the four Markov chain-based algorithms for it.  In Section \ref{trianglesection}, we give a simple way of characterizing the ordering returned by MC4 that will help with our analysis; we then show that both the Copeland and MC4 algorithms can be twice the optimal solution on a three element example.  We show the bad examples for MC1, MC2, and MC3 in Section \ref{mcsection}, and the family of bad examples for MC4 in Section \ref{mc4section}.  We give our 11-approximation for the Copeland score in Section \ref{copelandsection}.  We show the connection between Copeland and MC4 in Section \ref{mc4gensection}, and conclude in Section \ref{sec:conc}.

\section{Setting and Notation}
\label{sec:notation}

For consistency with Coppersmith et al.\ \cite{CFR} we denote $\{0,..,n-1\}$ as $[n]$. We only consider the \emph{full rank aggregation problem}, in which we are given permutations $\pi_1,\ldots,\pi_m:[n]\to[n]$ and the goal is to output a permutation $\sigma$ that minimizes the cost
\[c_R(\sigma)=\sum_{i=1}^m K(\sigma,\pi_i)\text{, where }
K(\sigma,\pi_i):=|\{(u,v)\in[n]^2: u<v, \sigma(u)<\sigma(v), \pi_i(u)>\pi_i(v)\}|\]
denotes the Kendall distance --- or Kendall tau distance --- between permutations $\sigma$ and $\pi_i$; the distance is simply the number of pairs of elements such that the ordering of the pair is different in $\sigma$ and $\pi_i$.  A permutation $\sigma$ will be written as a row vector, in which the $i$th coordinate holds the value $\sigma^{-1}(i)$. We will say that $\sigma^{-1}(i)$ is the $ith$ highest element in the ranking $\sigma$. In particular, this means that we will consider $i$ ranked higher than $j$ by $\sigma$ if $\sigma(i)<\sigma(j)$.

In the weighted feedback arc set problem, we are given a complete graph $G=(V,A)$ with weights $w_{ij},w_{ji}$ for each $i,j\in V$. We say that weights $w_{ij}$ obey the probability constraints if $w_{ij}+w_{ji}=1, w_{ij}+w_{jk}\geq w_{ik}\;\forall i,j,k\in V$. The goal of the problem is to find a permutation $\sigma$ that minimizes $c_G(\sigma)=\sum_{i,j: \sigma(i)<\sigma(j)} w_{ji}$. The weighted feedback arc set problem with weights obeying the probability constraints is a generalization of the rank aggregation problem: we set $w_{ij}=\frac{1}{m}\sum_{k=1}^m\mathbf{1}\{\pi_k(i)<\pi_k(j)\}$, and it is easy to see that this choice of weights obeys the probability constraints.

We will write for $i,j\in[n]$ that $i\succ j$ if a majority of input permutations rank $i$ higher than $j$ and denote $P_i=\{j: i\succ j\}$.\footnote{If $m$ is even, it might happen that exactly half of the permutations rank $i$ above $j$ and half rank $j$ above $i$. In this case, we assume an arbitrary tie-breaking rule for $i \succ j$ or $j\succ i$.}    The Copeland score of element $i$ is defined as $|P_i|$ and sorting by Copeland score means in non-increasing order, with arbitrary tie-breaking. We will refer to the majority tournament $T$ as the tournament on $n$ vertices, where arc $(i,j)$ is directed from $j$ to $i$ if $j\in P_i$.

We now define informally, and then formally, the four Markov chain-based algorithms.  Each Markov chain is on the set of elements $[n]$ and is specified by giving an $n\times n$ transition matrix $P=(p_{ij})$; $p_{ij}$ is the probability of transitioning from state $i$ to state $j$.  We compute an ordering from $P$ by finding $x$ such that $xP=x$,
and then returning a permutation $\pi$, that fulfills $x_{\pi^{-1}(0)}\geq\ldots\geq x_{\pi^{-1}(n-1)}$.
Ties are broken arbitrarily.  Given the graph of the transition probabilities with arcs $A=\{(i,j): p_{ij} \neq 0\}$, Dwork et al.\ note that the set of strongly connected components form a directed acyclic graph (DAG).  If there is more than one strongly connected component, Dwork et al.\ state that one should compute the ranking from the Markov chain on a strongly connected component corresponding to a source node of the DAG, put these nodes first in the ordering, and recurse on the remaining strongly connected components of the DAG.  In what follows, we will assume that there is a single strongly connected component, and our examples will also have this form.

For MC1, if the Markov chain is in state $i$, then the next state is chosen uniformly from the multiset of all elements that some input ranking ranked at least as high as $i$ (that is, from $\bigcup_k \{j: \pi_k(j) \leq \pi_k(i)\}$.)  Its transition matrix $P_{MC1}$ is defined by:
\[
p_{ij}=\frac{y_{ij}}{\sum_{k=1}^m x_{ik}},
\]
where $y_{ij}:=|\{k:\pi_k(j)\leq \pi_k(i)\}|$ denotes the number of permutations in which $j$ is ranked above or equal to $i$ and $x_{ik}:=|\{r:\pi_k(r)\leq \pi_k(i)\}|$ is the number of elements ranked above or equal to $i$ in permutation $\pi_k$.

For MC2, if the Markov chain is in state $i$, then the next state is chosen by picking one of the $m$ input rankings uniformly at random -- suppose it is $\pi_k$ -- then picking an element $j$ uniformly from the set of elements that $\pi_k$ ranks at least as high as $i$.  Its transition matrix $P_{MC2}$ is given by:
\[p_{ij}=\sum_{k=1}^m \frac{\mathbf{1}_{\{\pi_k(j)\leq \pi_k(i)\}}}{x_{ik}m}.
\]

For MC3, if the Markov chain is in state $i$, then the next state is chosen by picking one of the $m$ input rankings uniformly at random -- again suppose it is $\pi_k$ -- then picking an element $j \in [n]$ uniformly at random.  If $\pi_k(j) < \pi_k(i)$, the chain transitions to state $j$, otherwise it stays in $i$.  The transition matrix $P_{MC3}$ is defined by:
\[p_{ij}=\frac{z_{ij}}{mn},\]
where $z_{ij}=y_{ij}$ if $i\neq j$ and $z_{ii}=\sum_{k=1}^m (n-x_{ik}+1)$.

For MC4, if the Markov chain is in state $i$, then the next state is chosen by picking $j \in [n]$ uniformly at random.  If $j \succ i$ (that is, a majority of the input rankings rank $j$ higher than $i$), the chain moves to state $j$, otherwise it stays in $i$.
The transition matrix $P_{MC4}$ is given by
\[p_{ij}=\begin{cases}
\frac{1}{n} \text{ if } i\in P_j\\
0 \text{ if } i\neq j,i\not\in P_j\\
1-\sum_{r:r\neq i}p_{ir}\text{ if } i=j.
\end{cases}
\]

\section{Preliminaries}\label{trianglesection}
We begin by proving a simple lemma to characterize the ordering returned by MC4.

\begin{lemma}\label{mc4charlemma}
If $xP_{MC4}=x$, then $x_j=\frac{\sum_{i\in P_j}x_i}{n-|P_j|-1}$
\end{lemma}
\begin{proof}
\[
x_j=\sum_i x_i p_{ij}=\sum_{i\in P_j}\frac{x_i}{n}+\big(1-\sum_{r:r\neq j}p_{jr}\big)x_j
=\sum_{i\in P_j}\frac{x_i}{n}+x_j\frac{|P_j|+1}{n},\]
implying the result. The third equality uses that
\[
\sum_{r:r\neq j}p_{jr}=\sum_{r:j\in P_r}\frac{1}{n}=\frac{1}{n}|\{r:j\in P_r\}|=\frac{1}{n}((n-1)-|P_j|)=1-\frac{|P_j|+1}{n}.
\]
\end{proof}

The following is a simple corollary of the lemma, which we will later use:

\begin{corollary}\label{strictlyabove}
If $P_i\supset P_j$, then MC4 ranks $i$ above $j$.
\end{corollary}

We will now see that on an input of just three elements, both MC4 and sorting by Copeland's score can be off by a factor of 2.

We characterize the possible inputs consisting of the 6 possible permutations of 3 elements, such that the majority tournament is a directed triangle. The Markov chain is then completely symmetric, implying that $x_0=x_1=x_2$. Thus, we force MC4 to order arbitrarily. Similarly, as $|P_0|=|P_1|=|P_2|=1$ Copeland has to sort arbitrarily as well. We will see that in this setting, two permutations can be as much as a factor of 2 away from each other. Consider

\[
\pi_1=\left (0,1,2\right ),
\;
\pi_2=\left (0,2,1\right ),
\pi_3=\left (1,2,0\right ),
\;
\pi_4=\left (1,0,2\right ),
\;
\pi_5=\left (2,0,1\right ),
\;
\pi_6=\left (2,1,0\right )
\;
\]

\begin{center}
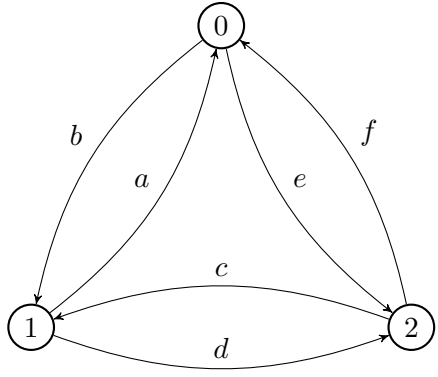
\begin{figure}
\begin{center}
\begin{tikzpicture}
[auto, ->,>=stealth',node distance=2cm and 2cm,vertex/.style={circle,draw=black,thick,inner sep=0pt,minimum size=6mm}]
\node[vertex] (0) at (0,0) {0};
\node[vertex] (1) at (-2.5, -4) {1};
\node[vertex] (2) at (2.5,-4) {2};
\path (0) edge[bend right=20] node[swap] {$b$} (1)
          edge[bend right=20] node {$e$} (2)
      (1) edge[bend right=20] node {$a$} (0)
          edge[bend right=20] node {$d$} (2)
      (2) edge[bend right=20] node[swap] {$f$} (0)
          edge[bend right=20] node[swap] {$c$} (1);
\end{tikzpicture}
\caption{Directed triangle}
\label{fig:triangle}
\end{center}
\end{figure}
\end{center}
The directed graph in Figure \ref{fig:triangle} yields a directed triangle as its majority tournament if either $a\geq b,c\geq d$ and $e\geq f$ or $a\leq b,c\leq d$ and $e\leq f$. We may assume without loss of generality that the first set of inequalities holds. If we let the ranking $\pi_i$ appear $\alpha_i$ times in the input, we can write
\begin{eqnarray*}
a=\alpha_3+\alpha_4+\alpha_6,\;\;\;c=\alpha_2+\alpha_5+\alpha_6,\;\;\;e=\alpha_1+\alpha_2+\alpha_4\\
b=\alpha_1+\alpha_2+\alpha_5, \;\;\;
d=\alpha_1+\alpha_3+\alpha_4,\;\;\;
f=\alpha_3+\alpha_5+\alpha_6.
\end{eqnarray*}

Substituting these into the previous inequalities, we find
\begin{eqnarray*}
\alpha_3+\alpha_4+\alpha_6\geq\alpha_1+\alpha_2+\alpha_5,\;\;
\alpha_2+\alpha_5+\alpha_6\geq
\alpha_1+\alpha_3+\alpha_4\\\text{ and }
\alpha_1+\alpha_2+\alpha_4\geq
\alpha_3+\alpha_5+\alpha_6.
\label{ineqs}
\end{eqnarray*}

Moreover, recalling that $c_R(\pi_i)$ denotes the cost of $\pi_i$ as output on the given input permutations, we find that
\begin{eqnarray*}
c_R(\pi_1)=\alpha_2+2\alpha_3+\alpha_4+2\alpha_5+3\alpha_6,\;\;\;
c_R(\pi_6)=3\alpha_1+2\alpha_2+\alpha_3+2\alpha_4+\alpha_5\\
c_R(\pi_2)=\alpha_1+3\alpha_3+2\alpha_4+\alpha_5+2\alpha_6,\;\;\;
c_R(\pi_3)=2\alpha_1+3\alpha_2+\alpha_4+2\alpha_5+\alpha_6,\\
c_R(\pi_4)=\alpha_1+2\alpha_2+\alpha_3+3\alpha_5+2\alpha_6,\;\;\;
c_R(\pi_5)=2\alpha_1+\alpha_2+2\alpha_3+3\alpha_4+\alpha_6
\end{eqnarray*}
Now, choosing e.g. $\alpha_2=1,\alpha_4=k,\alpha_6=k,\alpha_1=\alpha_3=\alpha_5=0$, we find that the inequalities are fulfilled - so there is a directed triangle and thus MC4 might output any permutation - and as $k\to\infty$, $\frac{c_R(\pi_1)}{c_R(\pi_6)}\to 2$. On the other hand, manipulating the inequalities above, we find that the ratio of the cost of any two permutations is at most 2, e.g.:
\[\frac{c_R(\pi_1)}{c_R(\pi_6)}=\frac{\alpha_2+2\alpha_3+\alpha_4+2\alpha_5+3\alpha_6}{3\alpha_1+2\alpha_2+\alpha_3+2\alpha_4+\alpha_5}\leq
\frac{(\alpha_2+\alpha_4)+\alpha_6+2(\alpha_3+\alpha_5+\alpha_6)}{2(\alpha_1+\alpha_2+\alpha_4)}
\leq\frac{1+1+2}{2}.\]

\begin{lemma}\label{newcopelandlowerbound}
Copeland's method, MC4 or any other algorithm whose output depends solely on the majority tournament, can be off by a factor of 2 if the input consists of permutations on 3 elements.
\end{lemma}

\section{MC1, MC2 and MC3}\label{mcsection}
We will show in this section that MC1, MC2 and MC3 all return a cost linear in $k$ on an input consisting of one copy of $\pi_5$ and $k-1$ copies of $\pi_1$ --- with $\pi_1,\pi_5$ defined as in the previous section.
%

The respective transition matrices will then be given by

\[\centering
P_{MC1}=
\begin{bmatrix}
\frac{k}{k+1} & 0 & \frac{1}{k+1}\\
\frac{k}{2k+1} & \frac{k}{2k+1} & \frac{1}{2k+1}\\
\frac{k-1}{3k-2} & \frac{k-1}{3k-2} & \frac{k}{3k-2}\\
\end{bmatrix}
\]
\[
P_{MC2}=
\begin{bmatrix}
\frac{k-1}{k}+\frac{1}{2k} & 0 & \frac{1}{2k}\\
\frac{k-1}{2k}+\frac{1}{3k} & \frac{k-1}{2k}+\frac{1}{3k} & \frac{1}{3k}\\
\frac{k-1}{3k} & \frac{k-1}{3k} & \frac{k+2}{3k}\\
\end{bmatrix}
\]
\[
P_{MC3}=
\begin{bmatrix}
\frac{3k-1}{3k} & 0 & \frac{1}{3k}\\
\frac{1}{3} & \frac{2k-1}{3k} & \frac{1}{3k}\\
\frac{k-1}{3k} & \frac{k-1}{3k} & \frac{k+2}{3k}\\
\end{bmatrix}
\]

We now show that the stationary distribution for all three of these Markov processes has $x_2>x_1$.

\begin{lemma}
In the stationary distribution of $P_{MC1}$, $x_2>x_1$.
\end{lemma}
\begin{proof}
$x_1=\frac{k}{2k+1}x_1+\left(\frac{k-1}{3k-2}\right)x_2\implies \frac{k+1}{2k+1}x_1=\frac{k-1}{3k-2}x_2\implies x_1<x_2$ for $k\geq2$.
\end{proof}

\begin{lemma}
In the stationary distribution of $P_{MC2}$, $x_2>x_1$.
\end{lemma}
\begin{proof}
$x_1=\frac{3k-1}{6k}x_1+\frac{2k-2}{6k}x_2\implies x_1\left(\frac{1}{2}+\frac{1}{6k}\right)=x_2\left(\frac{1}{3}-\frac{1}{3k}\right)\implies x_1<x_2$ for $k\geq2$.
\end{proof}

\begin{lemma}
In the stationary distribution of $P_{MC3}$, $x_2>x_1$.
\end{lemma}
\begin{proof}
$x_1=\left(\frac{2k-1}{3k}\right)x_1+\frac{k-1}{3k}x_2\implies (k+1)x_1=(k-1)x_2\implies x_1<x_2$.
\end{proof}

\begin{theorem}
The approximation guarantee of MC1, MC2 and MC3 cannot be bounded in the number of permuted elements.
\end{theorem}
\begin{proof}
The cost of any permutation that orders 2 before 1 is $(k-1)$, while the cost of  $\pi_1=(0,1,2)$ is $2$.
\end{proof}
\section{MC4}\label{mc4section}
In this section, we prove that MC4 has only a linear approximation guarantee even when the majority tournament is strongly connected.

Consider for $n\in\mathbb{N},c\approx\frac{n}{4}$ the following rankings:
\begin{align*}
\pi_1 & =\left (0,1,2,\ldots,n-1\right),\\
\pi_2 & =\left (1,2,\ldots,n-1,0\right ),\\
\pi_3 & =\left (n-c,n-c+1,\ldots,n-1,0,1,\ldots,n-c-1\right )
\end{align*}

The Markov chain created by MC4 on an input consisting of $n$ copies of $\pi_1$, $n$ copies $\pi_2$ and one copy of $\pi_3$ is shown in Figure \ref{fig:mc4}. We claim that the optimum has cost $O((n+c)n)=O(n^2)$, while MC4 returns a solution with cost $\Omega(n^3)$. Consider first $c_R(\pi_1)$:
\[
c_R(\pi_1)=n(K(\pi_1,\pi_1)+K(\pi_1,\pi_2))+K(\pi_1,\pi_3)=n(0+(n-1))+c(n-c)\in O(n^2).
\]
On the other hand, we claim that MC4 returns a solution, 
in which the number of elements ranked below $n-1$ is linear in $n$, as will be shown in Lemma \ref{positionlemma}.
\begin{figure}
\begin{center}
\begin{tikzpicture}
[auto, ->,>=stealth',node distance=1cm and 1cm,vertex/.style={circle,draw=black,thick,inner sep=1pt,minimum size=6mm}]
\node[vertex] (0) {0};
\node[vertex,below=of 0] (1) {1};
\node[vertex,below=of 1] (2) {2};
\node[below=.5 of 2] (d1) {$\vdots$};
\node[vertex,below=.5 of d1] (nc1) {$n-c-1$};
\node[vertex,below=of nc1] (nc) {$n-c$};
\node[below=.5 of nc] (d2) {$\vdots$};
\node[vertex,below=.5 of d2] (n1) {$n-1$};
\path (1) edge (0)
      (2) edge (1)
          edge[bend right] (0)
      (nc1) edge[bend right=10] (2)
            edge[bend right] (1)
            edge[bend right=45] (0)
      (nc) edge (nc1)
           edge[bend left=35] (2)
           edge[bend left=45] (1)
      (n1) edge[bend right=10] (nc)
           edge[bend right=20] (nc1)
           edge[bend right] (2)
           edge[bend right=45] (1)
      (0)  edge[bend right=55] (nc)
           edge[bend right=60] (n1);
\end{tikzpicture}
\end{center}
\caption{Markov chain created by MC4. Each edge has probability $\frac{1}{n}$ and with the remaining probability, there is a self-loop.}
\label{fig:mc4}
\end{figure}
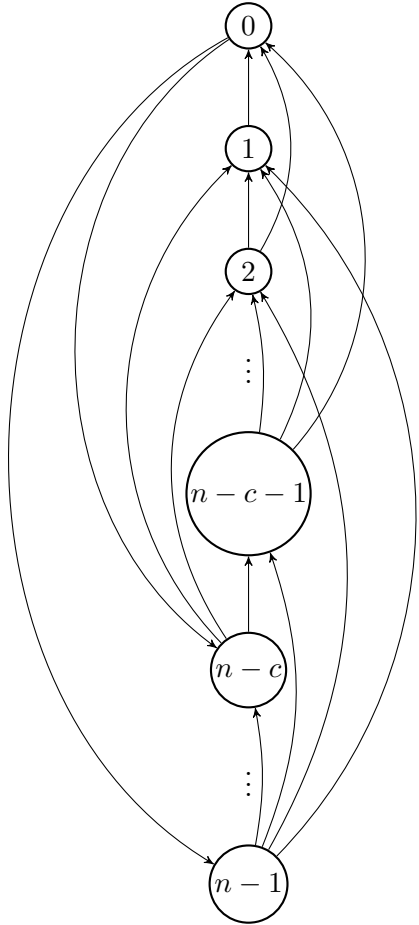
\begin{lemma}\label{positionlemma}
For the stationary vector $x$ of $P_{MC4}$: $x_{n-1}>x_k$ if $\sqrt{cn}+1<k<n-c$.
\end{lemma}
\begin{proof}
Consider the eigenvector $x$, such that $xP_{MC4}=x$. Notice that for $k:1\leq k \leq n-c-1$:
\[
x_k=\sum_{i:k<i\leq n-1}\frac{x_i}{n}+\frac{n-k}{n}x_k\implies x_k=\frac{\sum_{i:k<i}x_i}{k}.
\]
Additionally, we find that
\[
x_0=\sum_{i:1\leq i\leq n-c-1}\frac{x_i}{n}+\left(\frac{n-c}{n}\right)x_0\implies x_0=\sum_{i:1\leq i \leq n-c-1}\frac{x_i}{c}.
\]
and for $j\geq n-c$:
\[
x_j=\frac{x_0}{n}+\sum_{i:i>j}\frac{x_i}{n}+x_j\frac{n-(j-1)}{n}\implies x_j=\frac{x_0+\sum_{i:i>j}x_i}{j-1}>\sum_{i:1\leq i \leq n-c-1}\frac{x_i}{cn}.
\]
Here, the last inequality holds due to $j-1<n-1$ and the characterization of $x_0$ above.
We now want to compare the terms in
\[
x_{n-1}>\sum_{i:1\leq i \leq n-c-1}\frac{x_i}{cn}
\text{ and }
x_{k}=\sum_{i:k<i}\frac{x_i}{k}
\]
To do so, we count how often the sum $\sum_{i:k<i}\frac{x_i}{k}$ appears in $\sum_{i:1\leq i\leq n-c-1}x_i$:
\begin{align*}
\sum_{i:1\leq i \leq n-c-1}x_i
& =\sum_{i:2\leq i \leq n-c-1}2x_i+\sum_{i:i > n-c-1}x_i\\
& =\sum_{i:3\leq i \leq n-c-1}3x_i+\sum_{i:i > n-c-1}2x_i\\
& \ldots \\
& =\sum_{i:k\leq i \leq n-c-1}k x_i+\sum_{i:i > n-c-1}(k-1)x_i\\
& =\sum_{i:k<i\leq n-c-1}(k+1)x_i+\sum_{i:i > n-c-1}k x_i>k \sum_{i:k< i}x_i.
\end{align*}
But then, $x_{n-1}>k\sum_{i:k<i}\frac{x_i}{cn}>\sum_{i:k<i}\frac{x_i}{k-1}>x_k$ for $\sqrt{cn}+1<k\leq n-c-1$.\end{proof}
It follows from Corollary \ref{strictlyabove} that the elements $n-c,\ldots,n-2$ will all be ranked above $n-1$. Thus, the distance from $\pi_{MC4}$ to $\pi_1$ is at least $\Omega(c(n-c))=\Omega(n^2)$ giving a cost of $\Omega(n^3)$, which yields the following.
\begin{theorem}
MC4 has no sublinear approximation guarantee.
\end{theorem}

We will now show that if the majority tournament is strongly connected, then any solution is an $O(n)$-approximate solution.  We first need the following lemma.

\begin{lemma}
If the majority tournament is strongly connected, then OPT has cost $\Omega(mn)$.
\end{lemma}
\begin{proof}
Consider any permutation $\sigma$ and denote by $\alpha_i=\sigma^{-1}(i)$ the element ranked $i$th by $\sigma$. We may assume that $\alpha_0\succ\alpha_1\succ\cdots\succ\alpha_{n-1}$, as else swapping $\alpha_i,\alpha_{i+1}$ would decrease the cost of the permutation. However, strong connectivity implies that there exists a directed path $\alpha_0=\alpha_{i_0},\alpha_{i_1},\ldots,\alpha_{i_t}=\alpha_{n-1}$ in the majority tournament. Suppose first that $t=1$: then there are at least $\frac{m}{2}$ permutations in which $\alpha_{n-1}$ is ranked above $\alpha_0$. This would imply that the distance to $\sigma$ is at least 1 in these permutations. However, each of the rankings also includes $\alpha_1,\ldots,\alpha_{n-2}$ and ranks each such $\alpha_i$ either above $\alpha_0$, below $\alpha_{n-1}$ or both. In either case, each of these elements, together with at least one of $\alpha_0$ and $\alpha_{n-1}$, induce a distance of at least $1$ to $\sigma$. This adds up to a combined distance of $(n-1)$, giving a cost for $\sigma$ of at least $\frac{m(n-1)}{2}$.

For arbitrary $t$, we can make the same argument: if $i_j<i_{j+	1}$, then in each of the $\frac{m}{2}$ permutations, in which $\alpha_{i_j}$ is ranked below $\alpha_{i_{j+1}}$, each $\alpha_k\in\{\alpha_{i_j+1},\ldots,\alpha_{i_{j+1}}\}$ is ranked either above $\alpha_{i_j}$ or below $\alpha_{i_{j+1}}$, inducing a cost for $\sigma$ of at least $\frac{m}{2}(i_{j+1}-i_j)$. Summing over all such $j$, we find that $\sigma$ has a cost of at least $\frac{m(n-1)}{2}$.
%
\end{proof}

\begin{theorem}
With the assumption that the majority tournament is strongly connected, any solution is an $O(n)$-approximation guarantee.
\end{theorem}
\begin{proof}
The result is straightforward, since the distance between any two permutations is at most $\binom{n}{2}=O(n^2)$ and thus any solution returned can incur cost at most $m\binom{n}{2}$.
\end{proof}

\begin{corollary}
MC4 is a $O(n)$-approximation for the full rank aggregation problem.
\end{corollary}

\section{Sorting by Copeland score}\label{copelandsection}

In this section, we will show that sorting by Copeland score gives a constant factor approximation guarantee. To do so we will utilize a result by Coppersmith et al.\ \cite{CFR} proving that sorting by Borda scores gives a 5-approximation to the weighted feedback arc set problem as well as two lemmas from their proof.

Recall that $c_R(\sigma)$ and $c_G(\sigma)$ denote the cost of permutation $\sigma$ for the the rank aggregation problem and the corresponding FAS-problem on $G$, respectively.   We will  also need some additional notation: let $c_T(\sigma)$ denotes the cost of permutation $\sigma$ when solving the FAS-problem on the majority tournament $T$. $OPT_R,OPT_G, OPT_T$ will denote the respective optimal permutations for the rank-aggregation, the corresponding FAS-problem and the FAS-problem on the majority tournament respectively.

\begin{lemma}[Coppersmith et al.\ \cite{CFR}]\label{majoptbound}
$\forall\sigma\; c_T(\sigma)\leq 2c_G(\sigma)$. Notice that the inequality implies $c_T(OPT_T)\leq 2c_G(OPT_G)$.
\end{lemma}
\begin{proof}
We will prove the inequality by considering how much each edge of the majority tournament contributes to the respective cost functions.  An edge contributes either 0 or 1 to the left-hand side. Whenever it contributes 1 to the left-hand side exactly, it contributes at least $\frac{1}{2}$ to the right-hand side, and so the statement follows.
\end{proof}

\begin{lemma}[Coppersmith et al.\ \cite{CFR}]\label{majoptbound2}
$\forall \sigma\; c_T(\sigma)\geq c_G(\sigma)-c_G(OPT_G)$
\end{lemma}
\begin{proof}
Notice that an edge contributes either 0 or 1 to $c_T(\sigma)$. If it contributes $0$, then it contributes at most $\frac{1}{2}$ to $c_G(\sigma)$ and thus, at least as much to $c_G(OPT_G)$ as to $c_G(\sigma)$. Notice that if the edge contributes $1$ to $c_T(\sigma)$, the right-hand-side is bounded above by 1.
\end{proof}

\begin{theorem}
The approximation guarantee of sorting by Copeland score is at most 11.
\end{theorem}

\begin{proof}
We start by observing that $c_R(\sigma) = mc_G(\sigma)$ for any permutation $\sigma$, and thus $c_R(OPT_R) = mc_G(OPT_G)$, since the optimal solutions for the two problems are the same permutation.  Then
\begin{align*}
8c_R(OPT_R)& = 8mc_G(OPT_G)\\
& \geq 4m\big(c_T(OPT_T)\big)\\
& \geq m\big(c_T(Borda_T)-c_T(OPT_T)\big)\\
& =m \big(c_T(Copeland)-c_T(OPT_T)\big)\\
& \geq m\big(c_G(Copeland)-c_G(OPT_G)-c_T(OPT_T)\big)\\
& \geq m\big(c_G(Copeland)-c_G(OPT_G)-2c_G(OPT_G)\big)\\
& = c_R(Copeland)-3c_R(OPT_R)\qed
\end{align*}
Above, the first inequality follows from Lemma \ref{majoptbound} and the second from Lemma \ref{majoptbound2}. The second equality follows from the observation that Copeland and Borda score are equivalent with respect to the cost function $c_T$. The next inequality follows from Lemma \ref{majoptbound2}, and the last inequality follows from Lemma \ref{majoptbound}.
\end{proof}

\section{MC4 as a generalization of Copeland score}\label{mc4gensection}

The construction in Section \ref{mc4section} implies not only that there is no constant approximation guarantee for MC4, but also that there is no constant $c$ such that $|P_i|>c|P_j|\Rightarrow x_i>x_j$. For example, $x_{n-1}$ is ranked above $x_{\sqrt{cn}+2}$, despite $|P_{\sqrt{cn}+2}|\in\Omega(n),|P_{n-1}|=1$. From this perspective, MC4 is indeed far away from being a generalization of sorting by Copeland score. However, some implementations of the algorithm, as that of Schalekamp and van Zuylen \cite{Anke}, implement MC4 with a restart probability $\delta\in[0,1)$. We will denote such implementation as MC4$_\delta$. The transition matrix of MC4$_\delta$ is given as:

\[p_{ij}=\begin{cases}
\frac{1}{n} \text{ if } i\in P_j\\
\frac{\delta}{n} \text{ if } i\neq j,i\not\in P_j\\
1-\sum_{k:i\neq k}p_{ik}\text{ if } i=j.
\end{cases}
\]

We will now show that this allows two ways in which MC4 can yet be seen as a generalization of sorting by Copeland score. These will imply in particular that the linear lower bound  no longer applies when $\delta>0$. The constant lower bound in Lemma \ref{newcopelandlowerbound}, however, continues to hold for all $\delta>0$.

\begin{theorem}\label{MC4largedelta}
$\forall n,\forall\delta\in(1-\frac{1}{2n+1},1)$, a permutation output by MC4$_\delta$ can also be output when sorting by Copeland score. In particular, we then have that $|P_i|>|P_j|\implies x_i>x_j$, implying that MC4$_\delta$ is a constant-factor approximation.
\end{theorem}
\begin{proof}
We will use $1-\sum_{i:i\neq j}x_i=x_j$ and, as in the argument in Lemma \ref{mc4charlemma}, $\sum_{i:j\neq i}p_{ji}=\frac{\delta|P_j|+n-1-|P_j|}{n}=\frac{(n-1)+|P_j|(\delta-1)}{n}$
,  to write $x_j$ as
\begin{multline*}
x_j=\sum_ix_ip_{ij}=\sum_{i\in P_j}\frac{x_i}{n}(1-\delta)+\sum_{i:i\neq j}\frac{x_i}{n}\delta +x_j(1-\sum_{i:j\neq i}p_{ji})
\\
\implies x_j=\frac{\sum_{i\in P_j}x_i(1-\delta)+(1-x_j)\delta}{(n-1)+|P_j|(\delta-1)}=\frac{\delta-\delta x_j+\sum_{i\in P_j}x_i(1-\delta)}{(n-1)+|P_j|(\delta-1)}\\
\implies x_j\left(\frac{(n-1)+|P_j|(\delta-1)+\delta}{(n-1)+|P_j|(\delta-1)}\right)=\frac{(1-\delta)\sum_{i\in P_j}x_i+\delta}{(n-1)+|P_j|(\delta-1)}\\
\implies x_j=\frac{(1-\delta)\sum_{i\in P_j}x_i+\delta}{(n-1)+|P_j|(\delta-1)+\delta}=\frac{\delta+(1-\delta)\sum_{i\in P_j}x_i}{n-(1-\delta)(|P_j|+1)}
\end{multline*}

But then, it is easy to see that as $\delta\to1$, $x_j\to\frac{1}{n}$. In particular, it can be shown that with $\delta>\frac{2n}{2n+1}$, it holds that $|x_j-\frac{1}{n}|<\epsilon=\frac{1}{2n^2}$. But then, $|P_i|>|P_j|$ implies
\[
x_i>\frac{\delta+(1-\delta)(\frac{1}{n}-\epsilon)|P_i|}{n-(1-\delta)(|P_i|+1)}
\]and
\[
x_j<\frac{\delta+(1-\delta)(\frac{1}{n}+\epsilon)|P_j|}{n-(1-\delta)(|P_j|+1)} \leq
\frac{\delta+(1-\delta)(\frac{1}{n}+\epsilon)|P_j|}{n-(1-\delta)(|P_i|+1)}\leq
\frac{\delta+(1-\delta)(\frac{1}{n}-\epsilon)|P_i|}{n-(1-\delta)(|P_i|+1)} < x_i,
\] where the penultimate inequality holds
if $(\frac{1}{n}+\epsilon)|P_j|<(\frac{1}{n}-\epsilon)|P_i|$.
As $|P_j|<|P_i|\leq n-1$, the bound $\epsilon\leq \frac{1}{2n^2}$ implies $\epsilon(2|P_j|+1)<\frac{1}{n}$. Equivalently, $\epsilon |P_j| +\epsilon < \frac{1}{n}-\epsilon |P_j|$ or $(\frac{1}{n}+\epsilon)|P_j| < (\frac{1}{n}-\epsilon)(|P_j|+1)$. Since $|P_i|\geq |P_j|+1$, this implies the desired inequality.
\end{proof}

However, even for arbitrary $\delta>0$, we will now show that there is a way in which MC4 can be seen as a generalization of Copeland score:

\begin{theorem}
$\forall\delta>0\;\forall i,j: |P_i|>|P_j|\times\frac{1}{\delta^2}\implies x_i>x_j$, where $x_i,x_j$ result from the stationary distribution of MC4$_\delta$.
\end{theorem}
\begin{proof}
From the above expression of $x_i$ we can see that $x_i\geq \frac{\delta}{n}\;\forall i$. Similarly, using the expression and the facts that $\sum_k x_k=1$ and $|P_l|\leq n-1\;\forall l$, it can be shown that $x_i\leq \frac{1}{\delta n}\;\forall i$:
\[
x_i=\frac{\delta+(1-\delta)\sum_{k\in P_i}x_k}{n-(1-\delta)(|P_i|+1)}
\leq \frac{\delta+(1-\delta)}{n-(1-\delta)(|P_i|+1)}\leq\frac{1}{n-(1-\delta)n}=\frac{1}{\delta n}.
\]
Using those two bounds, we can then conclude that with $|P_i|>|P_j|\times \frac{1}{\delta^2}$:
\[
x_i=\frac{\delta+(1-\delta)\sum_{k\in P_i}x_k}{n-(1-\delta)(|P_i|+1)}
\geq \frac{\delta+(1-\delta)|P_i|\frac{\delta}{n}}{n-(1-\delta)(|P_i|+1)}
>\frac{\delta+(1-\delta)|P_i|\frac{\delta}{n}}{n-(1-\delta)(|P_j|+1)}\text{ and}
\]
\[
x_j=\frac{\delta+(1-\delta)\sum_{k\in P_j}x_k}{n-(1-\delta)(|P_j|+1)}\leq \frac{\delta+(1-\delta)|P_j|\frac{1}{\delta n}}{n-(1-\delta)(|P_j|+1)}
< \frac{\delta+(1-\delta)|P_i|\frac{\delta}{n}}{n-(1-\delta)(|P_j|+1)},
\]
which concludes the proof.
\end{proof}

\section{Conclusion}
\label{sec:conc}

We have shown that although MC4 performs extremely well in practice, its worst-case performance can be very bad.  This statement is similar to those made about other algorithms (for example, the running-time of the simplex method for linear programming).  It would be interesting to give some analytical explanation for the strong performance of MC4.  Perhaps one can characterize the types of input instances typically seen in practice, or perhaps a variant of the smoothed analysis of Spielman and Teng \cite{ST} used for the simplex method can be shown.  Another possible direction would be to characterize the settings of $\delta$ for which MC4 beats Copeland scoring, and thus has a constant performance guarantee.





\begin{thebibliography}{50}
\bibitem{Ailon} Ailon, Nir. ``Aggregation of partial rankings, p-ratings and top-m lists.'' Algorithmica 57.2 (2010): 284-300.
\bibitem{AilonCharikar} Ailon, Nir, and Moses Charikar. ``Fitting tree metrics: Hierarchical clustering and phylogeny.'' SIAM Journal on Computing 40.5 (2011): 1275-1291.
\bibitem{Newman} Ailon, Nir, Moses Charikar, and Alantha Newman. ``Aggregating inconsistent information: ranking and clustering.'' Journal of the ACM  55.5 (2008): 23.
\bibitem{PageRank} Brin, Sergey and Larry Page.  ``The anatomy of a large-scale hypertextual Web search engine.'' Computer Networks and ISDN Systems 30 (1998): 107-117.
\bibitem{CFR}Coppersmith, Don, Lisa Fleischer, and Atri Rudra. ``Ordering by weighted number of wins gives a good ranking for weighted tournaments.'' ACM Transactions on Algorithms 6.3 (2010): 55.
\bibitem{Dwork} Dwork, Cynthia, Ravi Kumar, Moni Naor, and Dandapani Sivakumar. ``Rank aggregation methods for the web.'' In Proceedings of the Tenth International Conference on World Wide Web, pp. 613-622. ACM, 2001.
\bibitem{FKMSV} Fagin, Ronald, Ravi Kumar, Dandapani Sivakumar, and Erik Vee. ``An Algorithmic View of Voting'', Manuscript.
\bibitem{Claire} Kenyon-Mathieu, Claire, and Warren Schudy. ``How to rank with few errors.'' Proceedings of the Thirty-Ninth Annual ACM Symposium on Theory of Computing, pp. 95-103. ACM, 2007.
\bibitem{Anke} Schalekamp, Frans, and Anke van Zuylen. ``Rank Aggregation: Together We're Strong.'' Proceedings of the Eleventh Workshop on Algorithm Engineering and Experiments (ALENEX), pp. 38-51. SIAM, 2009.
\bibitem{ST} Spielman, Daniel A. and Shang-Hua Teng.  ``Smoothed Analysis of Algorithms: Why the Simplex Method Usually Takes Polynomial Time.'' Journal of the ACM 51.3 (2004): 385-463.
\bibitem{Anke2} Van Zuylen, Anke, and David P. Williamson. ``Deterministic pivoting algorithms for constrained ranking and clustering problems.'' Mathematics of Operations Research 34.3 (2009): 594-620.

\end{thebibliography}


\end{document}